\documentclass[]{interact}

\usepackage{epstopdf}
\usepackage[caption=false]{subfig}
\usepackage{amssymb,amscd,amsfonts,amsmath,amsthm}
\usepackage{graphicx}
\usepackage{skak}
\usepackage{youngtab}
\usepackage[dvipsnames]{xcolor}
\usepackage{mathrsfs}
\usepackage{url}
\usepackage{color}
\usepackage{multirow}
\usepackage{hyperref}
\usepackage{xcolor, colortbl}
\usepackage{lscape}
\usepackage{array, multirow, multicol}
\definecolor{skyblue6}{rgb}{.2, .6, .8}
\definecolor{iceberg}{rgb}{0.44, 0.65, 0.82}
\definecolor{firebrick}{rgb}{.7, .13, .13}
\definecolor{blueice}{rgb}{.85, .96, .94}
\definecolor{lightcopper}{rgb}{.93, .76, .58}
\definecolor{buff}{rgb}{0.94, 0.86, 0.51}
\definecolor{celadon}{rgb}{0.67, 0.88, 0.69}
\usepackage[all]{xy}
\usepackage{float}

\usepackage[numbers,sort&compress]{natbib}
\bibpunct[, ]{[}{]}{,}{n}{,}{,}
\makeatletter
\def\NAT@def@citea{\def\@citea{\NAT@separator}}
\makeatother

\theoremstyle{plain}
\newtheorem{theorem}{Theorem}[section]

\newtheorem{corollary}[theorem]{Corollary}
\newtheorem{proposition}[theorem]{Proposition}
\newtheorem{conjetura}{Conjeture}

\theoremstyle{definition}
\newtheorem{definition}[theorem]{Definition}
\newtheorem{example}[theorem]{Example}

\theoremstyle{remark}
\newtheorem{remark}{Remark}

\begin{document}


\title{Construction of observable and MDP convolutional codes with good decodable properties for erasure channels by I/S/O representations \footnote{An initial version of this work was presented in ALAMA 2022}}

\author{
\name{N. DeCastro-Garc\'ia\textsuperscript{a}\thanks{CONTACT Noem\'i DeCastro-Garc\'ia. Email: ncasg@unileon.es}, M. V. Carriegos\textsuperscript{a} and A.L. Mu\~noz Casta\~neda\textsuperscript{a}}
\affil{\textsuperscript{a} Department of Mathematics. Universidad de Le\'on. Campus de Vegazana s/n 24071 Le\'on (Spain)}
}

\maketitle

\begin{abstract}
This paper addresses the construction of observable convolutional codes that exhibit good performance with the available decoding algorithms for erasure channels. Our construction is based on the use of input/state/output (I/S/O) representations and the invariance of certain properties of linear systems under various group actions. 

\end{abstract}



\section{Introduction}

Since Claude Shannon proposed entropy as a measure of information, coding theory has become an increasingly important field of research. In particular, convolutional codes are error-detecting and error-correcting codes used to transmit, detect, and correct information sent through a communication channel.

A fundamental issue in convolutional coding theory is the development of methods for constructing convolutional codes with desirable properties, such as error non-propagation (non-catastrophic behavior) and good performance when a decoding algorithm is applied. 
To this end convolutional coding theory is studied from different perspectives. 
In this article, we focus on the representation that models the dynamics of the convolutional code, describing it as a free submodule $\mathcal{C} \subset \mathbb{F}[z]^{n}$ of rank 
k. In this context, the associated linear dynamical system is known as an \emph{input/state/output} (I/S/O) representation. For this representation, it is required that the codewords have finite support. The resulting I/S/O representation is a minimal realization of the convolutional code as a discrete and reachable linear system.

I/S/O representations are valuable because they allow us to exploit the structural properties of linear systems in the context of coding theory. One important application is the use of reachable and observable representations to obtain non-catastrophic convolutional codes\cite{napiso,pinto2017,NPP2017,napp2019,nuevonapp,climentpinto,climent,valencia,isabel3,noeconcat}. Additionally, these representations play a key role in the construction of algebraic decoding algorithms for transmitting information over erasure or noisy channels \cite{lieb,decodingros,paconuevo}.

On the other hand, another desirable property of a convolutional code is a good distance profile, which ensures an optimal recovery rate. Maximum distance profile (MDP) convolutional codes exhibit this property. Therefore, finding methods to construct MDP convolutional codes presents an interesting problem in coding theory. The first procedure was proposed in \cite{strongly2} and is based on the construction of proper and lower triangular Toeplitz superregular matrices. This led to several algebraic works that provide general constructions for MDP convolutional codes \cite{ALMEIDA2013,ALMEIDA2016,angelMDP}. In \cite{lieb}, the concept of using superregular matrices to construct MDP convolutional codes, which are decodable via the proposed algorithm in an erasure channel is explored using I/S/O representations. More recently, in \cite{chen}, a family of MDP convolutional codes is constructed based on the theory of skew polynomials.

The goal of this article is to provide an algebraic method for constructing observable convolutional codes with desirable decoding properties for existing algorithms designed for erasure channels. The approach involves considering group actions on I/S/O representations, which allow us to generate new convolutional codes from a given one. The main advantage of these group actions is that they offer a new way to construct convolutional codes by applying well-established procedures that have proven effective in terms of distance and decoding properties.

This article is organized as follows: Section 2 presents an overview of the preliminary results. In Section 3, we prove our main results. Finally, the conclusions and references are provided.
\section{Preliminary Results}

In this section, we provide an overview of all the preliminary results. In classical convolutional coding theory, the input alphabet channel is the finite field \( \mathbb{F}_{q} \). From now on, we will denote it by \( \mathbb{F} \), assuming that \( q = p^{r} \).

\subsection{Convolutional Codes}

Convolutional codes were introduced by Peter Elias in 1955 \cite{elias}, and the first algebraic theoretical approach to convolutional codes was provided by Forney in \cite{forney}. There are several definitions of convolutional codes, depending on the specific requirements of the message transmission process. We consider a \( (n,k) \) convolutional code \( \mathcal{C} \) over a finite field \( \mathbb{F} \) as a \( \mathbb{F}[z] \)-generated finite submodule of rank \( k \), i.e., \( \mathcal{C} \subseteq \mathbb{F}[z]^{n} \) (Definition D', \cite{Conections}). The code rate is defined as \( \frac{k}{n} \).

The generator matrix \( G(z) \) of a \( (n,k) \) convolutional code \( \mathcal{C} \) over \( \mathbb{F}[z] \) is a $\mathbb{F}[z]$-linear map \( G(z): \mathbb{F}[z]^{l} \longrightarrow \mathbb{F}[z]^{n} \), \( u(z) \mapsto v(z) = G(z) \cdot u(z) \), such that \( \text{Im} \, G(z) = \mathcal{C} \). An encoder \( G(z) \) of a \( (n,k) \) convolutional code \( \mathcal{C} \) over \( \mathbb{F}[z] \) is a generator matrix with \( l = k \) and \( G(z) \) injective.

A given convolutional code admits several different convolutional encoders. We say that two encoders are equivalent if they generate the same code. It is well known that two encoders \( G(z) \) and \( G'(z) \) generate the same code if and only if there exists a unimodular matrix \( U(z) \in \mathbb{F}[z]^{k \times k} \) such that \( G'(z) = G(z) \cdot U(z) \) (Theorem 4, \cite{forney}).

Next, we recall a natural equivalence relation in the set of convolutional codes.
\begin{definition}
Two convolutional codes \( \mathcal{C}, \mathcal{C}' \subset \mathbb{F}[z]^{n} \) of dimension \( k \) are equivalent if and only if there exists a permutation of \( n \) letters, \( \sigma \in \mathcal{S}_{n} \), such that \( \sigma(\mathcal{C}) = \mathcal{C}' \). If \( G(z) \) and \( G'(z) \) are generator matrices for \( \mathcal{C} \) and \( \mathcal{C}' \), respectively, then \( \mathcal{C} \) and \( \mathcal{C}' \) are equivalent if and only if there exists a permutation matrix \( P \) and a unimodular matrix \( U(z) \in \mathbb{F}[z]^{k \times k} \) such that \( G'(z) = P \cdot G(z) \cdot U(z) \).
\end{definition}

An essential property of convolutional codes is  observability, which ensures that the code is not catastrophic. Algebraically, a \( (n,k) \) convolutional code \( \mathcal{C} \) over \( \mathbb{F} \) is observable if and only if there exists a syndrome map \( \psi:\mathbb{F}[z]^{n}\twoheadrightarrow \mathbb{F}[z]^{n-k} \) such that the following sequence is exact (Lemma 3.3.2, \cite{York}):
\[
0 \rightarrow \mathbb{F}[z]^{k} \overset{G(z)}{\longrightarrow} \mathbb{F}[z]^{n} \overset{\psi}{\longrightarrow} \mathbb{F}[z]^{n-k} \rightarrow 0.
\]
Let \( G(z) \in \mathbb{F}[z]^{n \times k} \) be an encoder of a \( (n,k) \) convolutional code and let $g_{ij}(z)$ denote the $(i,j)$ entry of $G(z)$. The following encoder properties are important.
\begin{enumerate}
\item The column degrees of the encoder are the \( k \)-integers \( \nu_{1}, \nu_{2}, \ldots, \nu_{k} \), where \( \nu_{j} = \max \left\{ \text{deg}(g_{ij}) \mid 1 \leq i \leq n \text{ and } j = 1, \ldots, k \right\} \) [cf. Definition 3.1.5, \cite{York}]. These are also called the constraint length of the \( j \)-th input of the matrix \( G(z) \).
\item The sum of the column degrees, \( \nu = \sum_{j=1}^{k} \nu_{j} \), is called the complexity of the encoder [cf. Definition 3.1.5, \cite{York}], the total memory or overall constraint length, or the external degree of the encoder.
\item The highest degree of the full-size minors \( k \times k \) of any encoder \( G(z) \), \( \delta(\mathcal{C}) \), is called the complexity of the convolutional code \( \mathcal{C} \) [cf. Definition 3.1.7, \cite{York}] or the internal degree of the encoder [cf. \cite{Mc2}].
\end{enumerate}
\begin{definition}[cf. Definition 2.25, \cite{Mc}]
\( G(z) \) is a minimal encoder if it is a generator matrix for which \( \nu \) is the smallest possible over all equivalent encoders.
\end{definition}
If \( G(z) \) is minimal, then \( \{\nu_{j}\}_{j=1 \ldots k} \) are called the Forney indices of the code, and \( \nu \) is the degree of the code, usually denoted by \( \delta \) [cf. Definition 2.25, \cite{Mc}]. When \( G(z) \) is minimal, \( \nu = \delta_{\mathcal{C}} \), the internal degree, is minimum [cf. Corollary 2.7, \cite{Mc2}, Definition 3.18, \cite{York}]. 

When working with convolutional codes over finite fields, we can always find minimal encoders. In this case, the Forney indices of the code are unique and invariant for the code. Therefore, \( \delta \) is an invariant of the code, and we say that \( \mathcal{C} \) is an \( (n,k,\delta) \)-convolutional code. In this case, the Forney indices are also called the Kronecker or controllability indices of the convolutional code.

Another important parameter in convolutional coding theory is the  distance of the code. The distance of a convolutional code is a measure of its ability to protect data from errors. Several types of distances can be defined for convolutional codes:

\begin{enumerate}
\item The Hamming distance between \( c_{1}, c_{2} \in \mathbb{F}^{n} \) is defined as \( d(c_{1}, c_{2}) = \text{wt}(c_{1} - c_{2}) \), where \( \text{wt}(c) \) is the Hamming weight of \( c \in \mathbb{F}^{n} \), which is the number of nonzero components of \( c \). Similarly, the distance between \( c_{1}(z), c_{2}(z) \in \mathbb{F}[z]^{n} \) is defined as \( d(c_{1}(z), c_{2}(z)) = \text{wt}(c_{1}(z) - c_{2}(z)) \), where \( \text{wt}(c(z)) = \sum_{t=0}^{\text{deg}(c(z))} \text{wt}(c_{t}) \), the sum of the Hamming weights of the coefficients of \( c(z) \).
\item The free distance of a convolutional code \( \mathcal{C} \) is the minimum Hamming distance between any two distinct code sequences and is given by
\[
d_{\textrm{free}}(\mathcal{C}) := \min_{c_{1}(z), c_{2}(z) \in \mathcal{C}} \left\{ d(c_{1}(z), c_{2}(z)) \mid c_{1}(z) \neq c_{2}(z) \right\}.
\]
\item In addition to the free distance, convolutional codes also have column distances (\cite{MDS3}). For \( j \in \mathbb{N}_{0} \), the \( j \)-th column distance of a convolutional code is defined as
\[
d_{j}^{c}(\mathcal{C}) := \min_{c(z) \in \mathcal{C}} \left\{ \text{wt}(c_{[0,j]}(z)) \mid c(z) \in \mathcal{C} \text{ and } c_0 \neq \mathbf{0} \right\},
\]
where \( c_{[0,j]}(z) = c_{0} + c_{1}z + \ldots + c_{j}z^{j} \) represents the \( j \)-th truncation of the code vector \( c(z) \in \mathcal{C} \).
\end{enumerate}

There are upper bounds for the free distance and column distances. Let \( \mathcal{C} \) be a \( (n,k,\delta) \)-convolutional code over \( \mathbb{F} \). Then, the following results hold (\cite{MDS3,MDS}):
\begin{enumerate}
\item \( d_{\textrm{free}}(\mathcal{C}) \leq (n-k) \left( \left\lfloor \frac{\delta}{k} \right\rfloor + 1 \right) + \delta + 1 \) (generalized Singleton bound).
\item \( d_{j}^{c}(\mathcal{C}) \leq (n-k)(j+1) + 1 \) for all \( j \in \mathbb{N}_{0} \).
\end{enumerate}
A \( (n,k,\delta) \)-convolutional code is called MDS (maximum distance separable) if the free distance equals the generalized Singleton bound. If the column distance equals its upper bound for \( j = 0, \ldots, L = \left\lfloor \frac{\delta}{k} \right\rfloor + \left\lfloor \frac{\delta}{n-k} \right\rfloor \), the convolutional code is called MDP (maximum distance profile) \cite{MDS3}.



\subsection{Convolutional Codes and Linear Systems}
\begin{definition}[cf. Section 5.2, \cite{York}, Remark 4.1, \cite{ros}]
Let $\mathcal{C} \subset \mathbb{F}[z]^n$ be a $(n, k, \delta)$-convolutional code. An Input/State/Output (I/S/O) representation of $\mathcal{C}$ is a tuple of matrices $\Sigma = (A \in \mathbb{F}^{\delta \times \delta}, B \in \mathbb{F}^{\delta \times k}, C \in \mathbb{F}^{(n-k) \times \delta}, D \in \mathbb{F}^{(n-k) \times k})$ such that
\begin{equation*}
\mathcal{C} =
\left\{
\begin{array}{l} 
v(z) = \left( \begin{array}{c} y(z) \\ u(z) \end{array} \right) \in \mathbb{F}[z]^n : \exists x(z) \in \mathbb{F}[z]^{\delta} \\
\textrm{ satisfying } \footnotesize \left\{ \begin{array}{rcl} 
x_{t+1} & = & A \cdot x_t + B \cdot u_t \\
y_t & = & C \cdot x_t + D \cdot u_t \\
v_t & = & \left( \begin{array}{c} y_t \\ u_t \end{array} \right) \\
x_0 & = & 0
\end{array} \right.
\end{array}
\right\}.
\end{equation*}
where $x_t \in \mathbb{F}^{\delta}$ denotes the state vector, $u_t \in \mathbb{F}^k$ represents the control/input/information vector, and $y_t \in \mathbb{F}^{n-k}$ is the output/parity vector for each time step $t$. The vector $v_t$ is the code vector transmitted over the communication channel.
\end{definition}
The existence of I/S/O representations for a given $(n, k, \delta)$-convolutional code over $\mathbb{F}$ is constructive, as described in \cite{York} and \cite{ros}, and it requires first obtaining a minimal first-order representation.
\begin{definition}[cf. Theorem 5.1.1, \cite{York}, Theorem 3.1, \cite{ros}]\label{first}
Let $\mathcal{C}\subset\mathbb{F}[z]^{n}$ be a $(n,k,\delta)$ convolutional code.
A triple of matrices, $(K, L, M)$, is called emph{first-order representation} of the code $\mathcal{C}$ if satisfy the following conditions:
\begin{enumerate}
\item[(i)] $\mathcal{C} = \{ v(z) \in \mathbb{F}[z]^n : \exists x(z) \in \mathbb{F}[z]^{\delta} \textrm{ such that } z K x(z) + L x(z) + M v(z) = 0 \}$.
\item[(ii)] $K$ is injective.
\item[(iii)] $(K, M)$ is surjective.
\end{enumerate}
Furthermore, if $(K, L, M)$ also satisfies the additional property
\begin{itemize}
    \item[(iv)] $(zK + L, M)$ is surjective,
\end{itemize}
then the first-order representation is said to be minimal.
\end{definition}
The triple $(K, L, M)$ is unique in the following sense: if $(\widehat{K}, \widehat{L}, \widehat{M})$ is another first-order representation of the convolutional code $\mathcal{C}$, then there exist unique and invertible matrices $T$ and $S$ of the appropriate sizes such that $(\widehat{K}, \widehat{L}, \widehat{M}) = (T K S^{-1}, T L S^{-1}, T M)$.

By appropriately permuting the codewords, if necessary, $(K, L, M)$ can be transformed via elementary transformations to obtain $\Sigma$:
\begin{equation*}
K' = \left( \begin{array}{c} 
-I_{\delta} \\ 
O 
\end{array} \right), \quad 
L' = \left( \begin{array}{c} 
A_{\delta \times \delta} \\ 
C_{(n-k) \times \delta} 
\end{array} \right), \quad 
M' = \left( \begin{array}{cc} 
O & B_{\delta \times k} \\
-I_{(n-k)} & D_{(n-k) \times k} 
\end{array} \right).
\end{equation*}
Moreover, for any $T \in GL_{\delta}(\mathbb{F})$, the equality 
\begin{equation*}\label{eq:equivalentiso}
\mathcal{C}({A}, {B}, {C}, {D}) = \mathcal{C}(T A T^{-1}, T B, C T^{-1}, D)
\end{equation*}
holds true, as shown in Proposition 2.2.8 of \cite{allen}.

\begin{remark}\label{rmk:equivalent}
It is implicit in the above construction that the system $\Sigma = (A, B, C, D)$ corresponds to an equivalence class (i.e., up to permutations) of codes. Therefore, if two convolutional codes are equivalent, then their I/S/O representations are also equivalent.
\end{remark}

Minimality is one of the most important properties of an I/S/O representation of a convolutional code, as it implies greater efficiency. An I/S/O representation of a convolutional code $\mathcal{C}$ is considered minimal if it is a reachable linear dynamical system.

\begin{definition}
Let $\Sigma = (A \in \mathbb{F}^{\delta \times \delta}, B \in \mathbb{F}^{\delta \times k}, C \in \mathbb{F}^{(n-k) \times \delta}, D \in \mathbb{F}^{(n-k) \times k})$ be a linear system over $\mathbb{F}$ of dimension $\delta$. $\Sigma$ is said to be reachable if its controllability matrix
\begin{equation*}
\Phi_{\delta}(A, B) = \left( \begin{array}{ccccc}
B & AB & \ldots & A^{\delta-2}B & A^{\delta-1}B
\end{array} \right)
\end{equation*}
has full rank, i.e., $\text{rank}(\Phi_{\delta}(A, B)) = \delta$.
\end{definition}

\begin{remark}[cf. Lemma 5.3.5, \cite{York}]
\label{remarema}
Note that if $(K, L, M)$ is minimal, the surjectivity of $(zK + L, M)$ implies that $\Sigma$ is a reachable system.
\end{remark}

Since we can obtain an encoder $G(z)$ by computing a minimal basis of the free $\mathbb{F}[z]$-module
\[
\text{Ker}[zK + L | M] = \left\{ v(z) \in \mathbb{F}[z]^n \mid \exists x(z) \in \mathbb{F}[z]^{\delta} : (zK + L) x(z) + M v(z) = 0 \right\},
\]
I/S/O representations allow us to construct observable convolutional codes based on properties of the associated system: if we start with a reachable and observable linear dynamical system, the resulting convolutional code is observable (Lemma 5.3.5, \cite{York}).

\begin{definition}
Let $\Sigma = (A \in \mathbb{F}^{\delta \times \delta}, B \in \mathbb{F}^{\delta \times k}, C \in \mathbb{F}^{(n-k) \times \delta}, D \in \mathbb{F}^{(n-k) \times k})$ be a linear system over $\mathbb{F}$ of dimension $\delta$. The system $\Sigma$ is observable if $\text{rank}(\Omega_{\delta}(A, C)) = \delta$, where $\Omega_{\delta}(A, C)$ is the observability matrix defined by
\begin{equation*}
\label{matrixobs}
\Omega_{\delta}(A, C) = \left( \begin{array}{c}
C \\
CA \\
CA^2 \\
\vdots \\
CA^{\delta-1}
\end{array} \right).
\end{equation*}
\end{definition}

Furthermore, if we have an observable I/S/O representation, there exist matrices $Y(z) \in \mathbb{F}[z]^{(n-k) \times k}$ and $U(z) \in \mathbb{F}[z]^{k \times k}$ such that 
$$G(z) = \left( \begin{array}{c} Y(z) \\ U(z) \end{array} \right)$$ 
is a generator matrix for $\mathcal{C}$. Then, the transfer function of the linear system $\mathcal{T}(z) = C(zI - A)^{-1}B + D$ is equal to $Y(z)U(z)^{-1}$, and thus (see Section 6 of \cite{Conections})
$$G(z) = \left( \begin{array}{c} \mathcal{T}(z) \cdot U(z) \\ U(z) \end{array} \right)$$

\subsection{The Decoding Process Using I/S/O Representations}

The I/S/O representations of convolutional codes provide an algebraic relationship between the input sequence \( u_t \) and the output sequence \( y_t \). For any code sequence \( v_t \), which must satisfy the dynamics of the system, if we have a \( (n,k,\delta) \)-convolutional code \( \mathcal{C} \subseteq \mathbb{F}[z]^n \) described by the system \( \Sigma \), the following equation must hold (Proposition 2.6, \cite{decoding}):

\begin{equation}
\label{sistemadecoding}
\left(\begin{array}{c} 
y_{t} \\
y_{t+1} \\
\vdots \\
y_{t+L} \end{array} \right) = 
\left(\begin{array}{c} 
C \\
CA \\
\vdots \\
CA^{L} \end{array} \right) \cdot x_{t} + 
\left(\begin{array}{ccccc} 
D & 0 & \cdots & 0 & 0 \\
CB & D & \cdots & \vdots & \vdots \\
CAB & CB & D & \cdots & \vdots \\
\vdots & \vdots & \cdots & CB & D \\
CA^{L-1}B & CA^{L-2}B & \cdots & CB & D \end{array} \right) 
\cdot \left(\begin{array}{c} 
u_{t} \\
u_{t+1} \\
\vdots \\
u_{t+L} \end{array} \right)
\end{equation}
where \( L := \left\lfloor \frac{\delta}{k} \right\rfloor + \left\lfloor \frac{\delta}{n-k} \right\rfloor > 0 \). The state evolution is modeled by the following equation:
\begin{equation*}
x_{\lambda} = A^{\lambda-t} x_t + \left(\begin{array}{cccc}
A^{\lambda-t-1}B & \ldots & AB & B
\end{array} \right)
\cdot \left(\begin{array}{c} 
u_t \\
u_{t+1} \\
\vdots \\
u_{t+L} 
\end{array} \right)
\end{equation*}
where \( \lambda = t+1, t+2, \ldots, t+L \). This system can also be written as:

\begin{equation*}
\label{sistemadecoding2}
\left(\begin{array}{c} 
C \\
CA \\
\vdots \\
CA^{L} \end{array} \right) \cdot x_t + 
\left(\begin{array}{c|ccccc} 
 & D & 0 & \cdots & 0 \\
 & CB & D & \cdots & \vdots \\
 -I & CAB & CB & \cdots & \vdots \\
 & \vdots & \vdots & \cdots & 0 \\
 & CA^{L-1}B & CA^{L-2}B & \cdots & CB & D 
\end{array} \right) 
\cdot \left(\begin{array}{c} 
y_t \\
y_{t+1} \\
\vdots \\
y_{t+L} \\
\vdots \\
u_t \\
u_{t+1} \\
\vdots \\
u_{t+L}
\end{array} \right) = 0
\end{equation*}
which we denote by:

\begin{equation}
\label{isoerasure}
\Omega_{L+1}(C,A) \cdot \mathbf{x_t} + [-I | F_L] \cdot \mathbf{v_t} = 0
\end{equation}
where \( F_L \) is a block Toeplitz matrix.

Let \( u(z) \) be an information word, and \( G(z) \) an encoder for a \( (n,k) \)-convolutional code. The codeword \( v(z) = G(z) \cdot u(z) \) is transmitted through the channel, potentially with errors. If no errors occur, \( v_t \) is a valid trajectory, and the error value is zero. Otherwise, if there is an error, the received sequence is not a codeword and does not belong to the code family. The goal of the decoding process is to recover \( v_t \) from the erroneous codeword \( \widehat{v_t} \). The most well-known decoding algorithm for convolutional codes is the Viterbi algorithm, as proposed in \cite{viterbi}. However, the existence of minimal I/S/O representations allows for the development of new decoding algorithms (e.g., \cite{lieb,decodingros,paconuevo,isabeldecoding}). When the channel is noisy,  the first decoding algorithms using I/S/O representations was introduced in \cite{decodingros}. In this context, the original information sequence \( u_t \), which has been corrupted by error, is recovered more efficiently than with the Viterbi algorithm, particularly for convolutional codes based on Reed-Solomon, BCH (\cite{BCH}), or strongly MDS codes (\cite{MDS2}). For erasure channels, the problem is not that the received message \( \widehat{v_t} \) is incorrect, but that \( v_t \) is incomplete due to lost symbols. For this case, an algorithm based on I/S/O representations was proposed in \cite{tomasvirtudes,tesisvirtudes} and improved in \cite{lieb}. The advantage of this algorithm is that it reduces the decoding delay and computational effort in the erasure recovery process.

This article focuses on decoding in erasure channels and recalls the main details of the decoding algorithm presented in \cite{tesisvirtudes} for such channels. Specifically, it considers the case where both the state and the symbols in \( v_t \) that are erased can be recovered simultaneously. The algorithm requires certain assumptions on the matrix \( (\Omega_{L+1}(A,C) \mid F_L) \). Let \( \alpha = (L+1)(n-k) - \delta \) and \( \mathcal{X} \leq \binom{(L+1)n}{\alpha} \). Assume that \( \{ j_1, \ldots, j_\alpha \} \) represents the set of indices of the columns of \( [-I \mid F_L] \) corresponding to the erased symbols in \( v_t \), where \( 1 \leq j_1 < j_2 < \ldots < j_\alpha \leq \mathcal{X} \). Also, suppose that \( \{ j_1, \ldots, j_\alpha \} \subset \{ i_1, i_2, \ldots, i_{(L+1)(n-k)} \} \) with \( i_{s(n-k)} \leq sn \) for \( s = 1, 2, \ldots, (L+1)(n-k) \).

Let \( U_{j_1, \ldots, j_\alpha} \) be the subspace spanned by the columns indexed by \( \{ j_1, \ldots, j_\alpha \} \), and let \( \langle \Omega_{L+1}(A,C) \rangle = \text{colspa}(\Omega_{L+1}(A,C)) \).

\begin{theorem}[cf. Theorem 4.3., \cite{tesisvirtudes}]
If for all such subspaces \( U_{j_1, \ldots, j_\alpha} \), the condition \( U_{j_1, \ldots, j_\alpha} \oplus \langle \Omega_{L+1}(A,C) \rangle = \mathbb{F}^{(L+1)(n-k)} \) holds, then it is possible to recalculate the state of the system (equation \ref{isoerasure}) when we observe a window of length \( (L+1)n \), provided that no more than \( (L+1)(n-k) - \delta \) erasures occur.
\end{theorem}

We denote by \( F_\alpha \) the matrix formed by the columns indexed by \( \{ j_1, \ldots, j_\alpha \} \). These columns are part of a non-trivially zero minor of the matrix \( [-I \mid F_L] \). 

\begin{definition}
\label{def:GDP}
A convolutional code is said to have the Good Decodable Property (GDP) for an erasure channel if it is constructed using an I/S/O representation such that \( [ \Omega_{L+1}(A,C) \mid F_\alpha ] \) is full rank for \( \alpha = (L+1)(n-k) - \delta \).
\end{definition}

As a summary, in order to construct observable and GDP convolutional codes using I/S/O representations, the following conditions must be satisfied:
\begin{enumerate}
\item[a)] The system \( \Sigma = (A,B,C,D) \) must be reachable.
\item[b)] The system \( \Sigma \) must be observable.
\item[c)] The columns of \( \Omega_{L+1}(A,C) \) and any \( \alpha \) columns of \( [-I \mid F_L] \) must be linearly independent, i.e., \( [ \Omega_{L+1}(A,C) \mid F_L ] \) must have full rank. Since \( \Sigma \) is observable, \( \Omega_{L+1}(A,C) \) has full rank, and if \( [\Omega_{L+1}(A,C) \mid F_L] \) is full rank, then \( [ \Omega_{L+1}(A,C) \mid F_\alpha ] \) is full rank.
\end{enumerate}

If \( \Sigma \) satisfies these conditions, the associated convolutional code will be observable and GDP. This article explores alternative methods for obtaining I/S/O representations that satisfy these conditions via group actions.

\section{Construction of observable and  GDP convolutional codes by I/S/O representations}
Let $\Sigma = (A \in \mathbb{F}^{\delta \times \delta}, B \in \mathbb{F}^{\delta \times k}, C \in \mathbb{F}^{(n-k) \times \delta}, D \in \mathbb{F}^{(n-k) \times k})$ be a minimal I/S/O representation of a convolutional code $\mathcal{C}$. By assumption, $\Sigma$ is reachable. We further assume that it is observable in order to facilitate the construction of observable convolutional codes. Additionally, $\Sigma$ satisfies the condition that $[\Omega_{L+1}(A,C) \mid F_{L}]$ is full rank, ensuring that $\mathcal{C}$ has the Good Decodable Property (GDP).

We now consider the following group transformations applied to $\Sigma$:
\begin{enumerate}
\item Group actions in the state vector: $\Sigma_1 = (A_1, B_1, C_1, D_1) = (S^{-1}AS, S^{-1}B, CS, D)$, where $S$ is an invertible matrix in $\mathbb{F}^{\delta \times \delta}$.
\item Group actions in the parity vector: $\Sigma_2 = (A_2, B_2, C_2, D_2) = (A, BQ, C, DQ)$, where $Q$ is an invertible matrix in $\mathbb{F}^{k \times k}$.
\item Group actions in the information vector: $\Sigma_3 = (A_3, B_3, C_3, D_3) = (A, B, H^{-1}C, H^{-1}D)$, where $H$ is an invertible matrix in $\mathbb{F}^{(n-k) \times (n-k)}$.
\end{enumerate}

We denote by $\mathcal{C}_i$ the associated convolutional code obtained by taking $\Sigma_i$ as the I/S/O representation, for $i = 1, 2, 3$.


\begin{proposition}
\label{reacha}
Let $\Sigma = (A \in \mathbb{F}^{\delta \times \delta}, B \in \mathbb{F}^{\delta \times k}, C \in \mathbb{F}^{(n-k) \times \delta}, D \in \mathbb{F}^{(n-k) \times k})$ be a linear system. The reachability of $\Sigma$ is invariant under the following group actions:
\begin{enumerate}
\item[i)] Group actions in the state vector: $\Sigma =(A, B, C, D) \mapsto \Sigma_{1}=(S^{-1}AS, S^{-1}B, CS, D)$ for some invertible matrix $S \in \mathbb{F}^{\delta \times \delta}$.
\item[ii)] Group actions in the parity vector: $\Sigma=(A, B, C, D) \mapsto \Sigma_{2}=(A, BQ, C, DQ)$ for some invertible matrix $Q \in \mathbb{F}^{k \times k}$. 
\item[iii)] Group actions in the information vector: $\Sigma=(A, B, C, D) \mapsto \Sigma_{3}=(A, B, H^{-1}C, H^{-1}D)$ for some invertible matrix $H \in \mathbb{F}^{(n-k) \times (n-k)}$. 
\end{enumerate}
\end{proposition}
\begin{proof}
We assume that $\Sigma = (A, B, C, D)$ is reachable, i.e., $\text{rank}(\Phi_{\delta}(A, B)) = \delta$.

(i) Since $\Phi_{\delta}(A_1, B_1) = S^{-1} \cdot \Phi_{\delta}(A, B)$, it follows that $\text{rank}(\Phi_{\delta}(A_1, B_1)) = \text{rank}(\Phi_{\delta}(A, B)) = \delta$.

(ii) Since $\Phi_{\delta}(A_2, B_2) = \Phi_{\delta}(A, B) \cdot \text{diag}(Q, Q, \ldots, Q)$, it follows that $\text{rank}(\Phi_{\delta}(A_2, B_2)) = \text{rank}(\Phi_{\delta}(A, B)) = \delta$.

(iii) Since $\Phi_{\delta}(A_3, B_3) = \Phi_{\delta}(A, B)$, it follows that $\text{rank}(\Phi_{\delta}(A_3, B_3)) = \text{rank}(\Phi_{\delta}(A, B)) = \delta$.
\end{proof}
\begin{corollary}
If $\Sigma$ is a minimal I/S/O representation of a convolutional code $\mathcal{C}$ over $\mathbb{F}$, then $\Sigma_{i}$ is an I/S/O representations of the convolutional code $\mathcal{C}_{i}$ for $i=1,2,3$ over $\mathbb{F}$.
\end{corollary}
\begin{proof}
It follows from that if $\Sigma$ is a minimal I/S/O representation of a convolutional code $\mathcal{C}$ over $\mathbb{F}$, then $\Sigma$ is reachable. Now, the result follows from Proposition \ref{reacha}.
\end{proof}

\begin{example}
\label{example1}
Let $\Sigma$ be the following system over $\mathbb{Z}_3$ with $\delta = 3$, $k = 2$, and $n = 3$:
\begin{equation*}
\Sigma = [A =
\begin{pmatrix}
0 & 1 & 0\\
2 & 1 & 0 \\
2 & 1 & 0
\end{pmatrix},
B = \begin{pmatrix}
0 & 0\\
0 & 2 \\
1 & 0
\end{pmatrix},
C = \begin{pmatrix}
1 & 1 & 2
\end{pmatrix},
D = \begin{pmatrix}
1 & 1
\end{pmatrix}].
\end{equation*}
The system $\Sigma$ is reachable because
\begin{equation*}
\Phi_3(A, B) = \begin{pmatrix}
B & AB & A^2B
\end{pmatrix} = \begin{pmatrix}
0 & 0 & 0 & 2 & 0 & 2 \\
0 & 2 & 0 & 2 & 0 & 2 \\
1 & 0 & 0 & 2 & 0 & 2
\end{pmatrix}
\end{equation*}
has full rank. Thus, we can take $\Sigma$ as a minimal I/S/O representation of a convolutional code. An associated encoder for the convolutional code $\mathcal{C}$ constructed by $\Sigma$ is
\begin{equation*}
G(z) = \begin{pmatrix}
1 + z + 2z^2 & 2 + z\\
1 + 2z & z \\
2 + z + 2z^2 & 0
\end{pmatrix}.
\end{equation*}
Let $Q = \begin{pmatrix}
1 & 1\\
1 & 2
\end{pmatrix}$ be an invertible matrix over $\mathbb{Z}_3$. We apply the group transformation $\Sigma_2 = (A_2, B_2, C_2, D_2) = (A, BQ, C, DQ)$ to $\Sigma$. Then,
\begin{equation*}
\Sigma_2 = [A_2 =
\begin{pmatrix}
0 & 1 & 0\\
2 & 1 & 0 \\
2 & 1 & 0
\end{pmatrix},
B_2 = \begin{pmatrix}
0 & 0\\
2 & 1 \\
1 & 1
\end{pmatrix},
C_2 = \begin{pmatrix}
1 & 1 & 2
\end{pmatrix},
D_2 = \begin{pmatrix}
2 & 0
\end{pmatrix}].
\end{equation*}
The transformed system $\Sigma_2$ is also reachable because
\begin{equation*}
\Phi_3(A_2, B_2) = \begin{pmatrix}
B_2 & A_2B_2 & A_2^2B_2
\end{pmatrix} = \begin{pmatrix}
0 & 0 & 2 & 1 & 2 & 1 \\
2 & 1 & 2 & 1 & 0 & 0 \\
1 & 1 & 2 & 1 & 0 & 0
\end{pmatrix}
\end{equation*}
has full rank. Thus, we can take $\Sigma_2$ as a minimal I/S/O representation to compute an encoder for the convolutional code $\mathcal{C}_2$,
\begin{equation*}
G_2(z) = \begin{pmatrix}
1 + z + z^2 & 2 + z\\
z^2 & 2z \\
1 + z + z^2 & 1
\end{pmatrix}.
\end{equation*}
Let $H = \begin{pmatrix}
2
\end{pmatrix}$ be an invertible matrix in $\mathbb{Z}_3^{1 \times 1}$. We apply the group transformation $\Sigma_3 = (A_3, B_3, C_3, D_3) = (A, B, H^{-1}C, H^{-1}D)$ to $\Sigma$. Then,
\begin{equation*}
\Sigma_3 = [A_3 =
\begin{pmatrix}
0 & 1 & 0\\
2 & 1 & 0 \\
2 & 1 & 0
\end{pmatrix},
B_3 = \begin{pmatrix}
0 & 0\\
0 & 2 \\
1 & 0
\end{pmatrix},
C_3 = \begin{pmatrix}
2 & 2 & 1
\end{pmatrix},
D_3 = \begin{pmatrix}
2 & 2
\end{pmatrix}].
\end{equation*}
The transformed system $\Sigma_3$ is also reachable because $\Phi_3(A_3, B_3) = \Phi_3(A, B)$ has full rank. Thus, we can take $\Sigma_3$ as a minimal I/S/O representation to compute an encoder for the convolutional code $\mathcal{C}_3$,
\begin{equation*}
G_3(z) = \begin{pmatrix}
2 + 2z + z^2 & 1 + 2z\\
1+2z & z \\
2 + z + 2z^2 & 0
\end{pmatrix}.
\end{equation*}
\end{example}
\begin{proposition}
\label{obs}
Let $\Sigma= (A \in \mathbb{F}^{\delta \times \delta}, B \in \mathbb{F}^{\delta \times k}, C \in \mathbb{F}^{(n-k) \times \delta}, D \in \mathbb{F}^{(n-k) \times k})$ be a linear dynamical system. The observability of $\Sigma$ is invariant under the following group actions:
\begin{enumerate}
\item[i)] Group actions in the state vector: $\Sigma =(A, B, C, D) \mapsto \Sigma_{1}=(S^{-1}AS, S^{-1}B, CS, D)$ for some invertible matrix $S \in \mathbb{F}^{\delta \times \delta}$.
\item[ii)] Group actions in the parity vector: $\Sigma=(A, B, C, D) \mapsto \Sigma_{2}=(A, BQ, C, DQ)$ for some invertible matrix $Q \in \mathbb{F}^{k \times k}$. 
\item[iii)] Group actions in the information vector: $\Sigma=(A, B, C, D) \mapsto \Sigma_{3}=(A, B, H^{-1}C, H^{-1}D)$ for some invertible matrix $H \in \mathbb{F}^{(n-k)\times (n-k)}$.
\end{enumerate}
\end{proposition}

\begin{proof}
Assume that $\Sigma=(A,B,C,D)$ is observable, i.e., $rank \, \Omega_{\delta}(A,C) = \delta$.

(i) Since $\Omega_{\delta}(A_{1},C_{1}) = \Omega_{\delta}(A,C) \cdot S$, we have $rank(\Omega_{\delta}(A_{1},C_{1})) = rank(\Omega_{\delta}(A,C)) = \delta$.

(ii) Since $\Omega_{\delta}(A_{2},C_{2}) = \Omega_{\delta}(A,C)$, it follows that $rank(\Omega_{\delta}(A_{2},C_{2})) = rank(\Omega_{\delta}(A,C)) = \delta$.

(iii) Since $\Omega_{\delta}(A_{3},C_{3}) = diag(H^{-1}, \ldots, H^{-1}) \cdot \Omega_{\delta}(A,C)$, we conclude that $rank(\Omega_{\delta}(A_{3},C_{3})) = rank(\Omega_{\delta}(A,C)) = \delta$.
\end{proof}
\begin{corollary}
If $\Sigma$ is a minimal and observable I/S/O representation of an observable convolutional code $\mathcal{C}$ over $\mathbb{F}$, $\Sigma_{i}$ for $i=1,2,3$ can be considered observable I/S/O representations for observable convolutional codes $\mathcal{C}_{i}$ over $\mathbb{F}$.
\end{corollary}
\begin{proof}
It follows from the fact that $\Sigma$ is an observable linear system.
\end{proof}

\begin{example}
Let $\Sigma$ and $\Sigma_{2}$ be the systems defined over $\mathbb{Z}_{3}$ as given in Example \ref{example1}.
The system $\Sigma$ is observable because 
\[
\Omega_{3}(A,C) = \left(\begin{array}{c}
C \\
CA \\
CA^{2}
\end{array}\right) = \left(\begin{array}{ccc}
1 & 1 & 2 \\
0 & 1 & 0 \\
2 & 1 & 0 
\end{array}\right)
\]
has full rank. Therefore, $\mathcal{C}$ is an observable convolutional code. Similarly, the system $\Sigma_{2}$ is also observable because 
\[
\Omega_{3}(A_{2},C_{2}) = \left(\begin{array}{c}
C_{2} \\
C_{2}A_{2} \\
C_{2}A_{2}^{2}
\end{array}\right) = \Omega_{3}(A,C).
\]
Thus, $\mathcal{C}_{2}$ is an observable convolutional code.

On the other hand, let $\Sigma$ and $\Sigma_{3}$ be the systems defined over $\mathbb{Z}_{3}$ as given in Example \ref{example1}. The system $\Sigma_{3}$ is observable because the observability matrix
\[
\Omega_{3}(A_{3},C_{3}) = \left(\begin{array}{c}
C_{3} \\
C_{3}A_{3} \\
C_{3}A_{3}^{2}
\end{array}\right) = \left(\begin{array}{ccc}
2 & 2 & 1\\
0 & 2 & 0\\
1 & 2 & 0
\end{array}\right)
\]
has full rank. Thus, $\mathcal{C}_{3}$ is an observable convolutional code.
\end{example}

The first question that arises is whether the convolutional codes $\mathcal{C}_{i}$ are equivalent to $\mathcal{C}$. In fact, this is not our primary concern, as we aim to obtain the largest possible number of distinct convolutional codes with good properties. Trivially, we know that when the first transformation is applied to the representation, the resulting I/S/O systems are equivalent by similarity. Consequently, the codes they generate are equal. However, this statement does not hold for transformations obtained through $\Sigma_{2}$ and $\Sigma_{3}$. Let us examine an example below.

\begin{example}
Let $\Sigma$ and $\Sigma_{2}$ be the I/S/O representations given in Example \ref{example1}. Let 
\[
G(z) = \left(\begin{array}{cc}
1 + z + 2z^2 & 2 + z \\
1 + 2z & z \\
2 + z + 2z^2 & 0
\end{array}\right)
\quad \textrm{and} \quad 
G_{2}(z) = \left(\begin{array}{cc}
1 + z + z^2 & 2 + z \\
z^2 & 2z \\
1 + z + z^2 & 1
\end{array}\right)
\]
be the encoders for $\mathcal{C}$ and $\mathcal{C}_{2}$, respectively, obtained from $\Sigma$ and $\Sigma_{2}$. 
Then, $\mathcal{C}$ is not equivalent to $\mathcal{C}_{2}$. Let us consider this. Let $v(z) \in \mathbb{Z}_{3}[z]^{3}$ be the codeword defined as 
$$v(z) = \left(\begin{array}{c}
1 + z + z^2 \\
z^2 \\
1 + z + z^2
\end{array}\right).$$ 
Clearly, $v(z) \in \mathcal{C}_{2}$, but $v(z) \notin \mathcal{C}$. If $v(z) \in \mathcal{C}$, then there would exist 
$$x(z) = \left(\begin{array}{c}
a(z) \\
b(z)
\end{array}\right) \in \mathbb{Z}_{3}[z]^{2}$$ 
such that 
\[
\left(\begin{array}{cc}
1 + z + 2z^2 & 2 + z \\
1 + 2z & z \\
2 + z + 2z^2 & 0
\end{array}\right)
\cdot \left(\begin{array}{c}
a(z) \\
b(z)
\end{array}\right)
= \left(\begin{array}{c}
1 + z + z^2 \\
z^2 \\
1 + z + z^2
\end{array}\right).
\]
Solving this equation, we find that $(2 + z + 2z^2) \cdot a(z) = 1 + z + z^2$, implying $a(z) = \lambda \in \mathbb{Z}_{3}$. However, $(2 + z + 2z^2) \cdot \lambda = 1 + z + z^2$ has no solution, so $v(z) \notin \mathcal{C}$. This can also be seen for words $\tilde{v(z)}$ obtained from $v(z)$ by permuting is components.
Thus, the codes are not equivalent.
\end{example}

\begin{example}
Let $\Sigma$ and $\Sigma_{3}$ be the I/S/O representations given in Example \ref{example1}. Let 
\[
G(z) = \left(\begin{array}{cc}
1 + z + 2z^2 & 2 + z \\
1 + 2z & z \\
2 + z + 2z^2 & 0
\end{array}\right)
\quad \textrm{and} \quad 
G_{3}(z) = \left(\begin{array}{cc}
2 + 2z + z^2 & 1 + 2z \\
1+2z & z \\
2 + z + 2z^2 & 0
\end{array}\right)
\]
be the encoders for $\mathcal{C}$ and $\mathcal{C}_{3}$, respectively, obtained from $\Sigma$ and $\Sigma_{3}$. Then, $\mathcal{C}$ is not equivalent to $\mathcal{C}_{3}$. Let us consider this. Let $v(z) \in \mathbb{Z}_{3}[z]^{3}$ be the codeword defined as 
$$v(z) = \left(\begin{array}{c}
2 + 2z + z^2 \\
1+2z \\
2 + z + 2z^2
\end{array}\right).$$ 
Clearly, $v(z) \in \mathcal{C}_{3}$, but $v(z) \notin \mathcal{C}$. If $v(z) \in \mathcal{C}$, then there would exist 
$$x(z) = \left(\begin{array}{c}
a(z) \\
b(z)
\end{array}\right) \in \mathbb{Z}_{3}[z]^{2}$$ 
such that 
\[
\left(\begin{array}{cc}
1 + z + 2z^2 & 2 + z \\
1 + 2z & z \\
2 + z + 2z^2 & 0
\end{array}\right)
\cdot \left(\begin{array}{c}
a(z) \\
b(z)
\end{array}\right)
= \left(\begin{array}{c}
2 + 2z + z^2 \\
1+2z \\
2 + z + 2z^2
\end{array}\right).
\]
Solving the above equation, we find that $(2 + z + 2z^2) \cdot a(z) = 2+z+2z^"$, implying $a(z) = 1 \in \mathbb{Z}_{3}$. Then, 
substituting into the equation $(1+2z)\cdot a(z)+z\cdot b(z)=1+2z$, we obtain that $b(z)=0$. But, with this values, the equation $(1+z+2z^2)\cdot a(z) + (2+z) \cdot b(z) \neq 2+2z+2z^2$. So, $v(z) \notin \mathcal{C}$.  This can also be seen for words $\tilde{v(z)}$ obtained from $v(z)$ by permuting is components.
Thus, the codes are not equivalent.
\end{example}

Due to the results obtained above, the possible equivalence bewteen the convolutional codes will depend on the transformation matrices $Q$ and $H$, but there will exist some properties on them that let us obtain no equivalent convolutional codes that keep the good properties. For this reason, we continue only with the group transformations on the parity and information vectors, as these are the transformations that can lead to different convolutional codes compared to the one we initially had. 

Recall that we want to explore possible methods to construct observable and GDP convolutional codes (see \cite{tesisvirtudes,lieb} for related works). From the systems theory perspective, the solvability of Eq. (\ref{sistemadecoding}) can be studied by analyzing the characterization of the I/S/O representation $\Sigma$ of the code as an \emph{output observable system} (\cite{isabel2}, \cite{isabeldecoding}, \cite{isabel}, \cite{isabelhijo}). For convenience, we denote the matrix $[\Omega_{L+1}(A,C) \mid F_{L}]$ by $T_{L}$.

\begin{definition}
A system $\Sigma$ is \emph{output observable} if the sequence of states $x(0), \ldots, x(L)$ is uniquely determined by the knowledge of the output sequence $y(0), \ldots, y(L)$ for a finite number of steps $L \in \mathbb{N}$. If the matrix $T_{L}$ is defined as follows:
\begin{equation}
\label{matrizdecoding}
T_{L} = \left( \begin{array}{cccccc}
C & D &  &  & & \\
CA & CB & D & \vdots & \vdots & \\
CA^{2} & CAB & CB & D & \vdots & \\
\vdots & \vdots & & \vdots & & \\
CA^{L} & CA^{L-1}B & CA^{L-2}B & \cdots & CB & D \end{array} \right),
\end{equation}
then $\Sigma$ is \emph{output observable} if and only if $rank(T_{L})$ is maximum for all $L \in \mathbb{N}$.
\end{definition}

\begin{proposition}
\label{main}
Let $\Sigma$ be a linear system over $\mathbb{F}$. The output observability of the system $\Sigma$ over $\mathbb{F}$ is invariant under the following group actions:
\begin{enumerate}
\item[i)] Group actions on the parity vector: $\Sigma = (A, B, C, D) \mapsto \Sigma_{2} = (A, BQ, C, DQ)$ for some invertible matrix $Q \in \mathbb{F}^{k \times k}$.
\item[ii)] Group actions on the information vector: $\Sigma = (A, B, C, D) \mapsto \Sigma_{3} = (A, B, H^{-1}C, H^{-1}D)$ for some invertible matrix $H \in \mathbb{F}^{(n-k)\times (n-k)}$.
\end{enumerate}
\end{proposition}

\begin{proof}
We denote by $T^{(i)}_{L}$ the matrix $[\Omega_{L+1}(A_{i}, C_{i}) \mid F^{i}_{L}]$ obtained from $\Sigma_{i}$. It suffices to prove that $rank[T_{L}] = rank[T^{(i)}_{L}]$, and then, if $\Sigma$ is output observable, we conclude that $\Sigma_{i}$ is also output observable.\\
(i) Since
\begin{equation*}
T^{(2)}_{L} = \left( \begin{array}{cccccc} 
C & DQ &  &  & & \\
CA & CBQ & DQ & \vdots & \vdots & \\
CA^{2} & CABQ & CBQ & DQ & \vdots & \\
\vdots & \vdots & & \vdots & & \\
CA^{L} & CA^{L-1}BQ & CA^{L-2}BQ & \cdots & CBQ & DQ \end{array} \right)
= 
\end{equation*}
\begin{equation*}
\resizebox{\textwidth}{!}{$\displaystyle
=\left( \begin{array}{cccccc} 
C & D &  &  & & \\
CA & CB & D & \vdots & \vdots & \\
CA^{2} & CAB & CB & D & \vdots & \\
\vdots & \vdots & & \vdots & & \\
CA^{L} & CA^{L-1}B & CA^{L-2}B & \cdots & CB & D \end{array} \right)
\cdot
\left( \begin{array}{ccccc} 
I & O &  &  & \\
O & Q & O & \vdots & \vdots \\
O & O & Q & O & \vdots \\
\vdots & \vdots & & \vdots & \\
O & O & O & \cdots & Q \end{array} 
\right) = T_{L} \cdot \mathcal{Q}$,
}
\end{equation*}
and since $\mathcal{Q}$ is invertible, we conclude that $rank[T_{L}] = rank[T^{(2)}_{L}]$.\\
(ii) Since

\begin{equation*}
T^{(3)}_{L} = \left( \begin{array}{cccccc} 
H^{-1}C & H^{-1}D &  &  & & \\
H^{-1}CA & H^{-1}CB & H^{-1}D & \vdots & \vdots & \\
H^{-1}CA^{2} & H^{-1}CAB & H^{-1}CB & H^{-1}D & \vdots & \\
\vdots & \vdots & & \vdots & & \\
H^{-1}CA^{L} & H^{-1}CA^{L-1}B & H^{-1}CA^{L-2}B & \cdots & H^{-1}CB & H^{-1}D \end{array} \right)
= 
\end{equation*}
\begin{equation*}
\resizebox{\textwidth}{!}{$\displaystyle
=\left( \begin{array}{ccccc} 
H^{-1} & O &  &  & \\
O & H^{-1} & O & \vdots & \vdots \\
O & O & H^{-1} & O & \vdots \\
\vdots & \vdots & & \vdots & \\
O & O & O & \cdots & H^{-1} \end{array} \right)
\cdot
\left( \begin{array}{cccccc} 
C & D &  &  & & \\
CA & CB & D & \vdots & \vdots & \\
CA^{2} & CAB & CB & D & \vdots & \\
\vdots & \vdots & & \vdots & & \\
CA^{L} & CA^{L-1}B & CA^{L-2}B & \cdots & CB & D \end{array} \right)
= \mathcal{H} \cdot T_{L},
$}
\end{equation*}
and since $\mathcal{H}$ is invertible, we conclude that $rank[T_{L}] = rank[T^{(3)}_{L}]$.
\end{proof}
\begin{corollary}
If $\Sigma$ is an I/S/O representation of a GDP convolutional code $\mathcal{C}$ over $\mathbb{F}$, $\Sigma_{i}$ for $i=1,2,3$ are I/S/O representations for GDP convolutional codes $\mathcal{C}_{i}$ over $\mathbb{F}$.
\end{corollary}
\begin{proof}
It follows from Definition \ref{def:GDP} and Proposition \ref{main}.
\end{proof}

\begin{example}
Let $\Sigma$ and $\Sigma_{2}$ be the systems obtained in Example \ref{example1}. The system $\Sigma$ is output observable because

\begin{equation*}
T_{L} = \left( \begin{array}{ccccccccccc} 
1 & 1 & 2 & 1 & 1 & 0 & 0 & 0 & 0 & 0 & 0 \\
0 & 1 & 0 & 2 & 2 & 1 & 1 & 0 & 0 & 0 & 0 \\
2 & 1 & 0 & 0 & 2 & 2 & 2 & 1 & 1 & 0 & 0 \\
2 & 0 & 0 & 0 & 2 & 0 & 2 & 0 & 2 & 1 & 1 \end{array} \right),
\end{equation*}
which has full rank. Also, $\Sigma_{2}$ is output observable because

\begin{equation*}
T^{(2)}_{L} = \left( \begin{array}{ccccccccccc} 
1 & 1 & 2 & 2 & 0 & 0 & 0 & 0 & 0 & 0 & 0 \\
0 & 1 & 0 & 1 & 0 & 2 & 0 & 0 & 0 & 0 & 0 \\
2 & 1 & 0 & 2 & 1 & 1 & 0 & 2 & 0 & 0 & 0 \\
2 & 0 & 0 & 2 & 1 & 2 & 1 & 1 & 0 & 2 & 0 \end{array} \right),
\end{equation*}
which also has full rank. Therefore, since $\mathcal{C}$ is GDP, the convolutional code $\mathcal{C}_{2}$ is also GDP. 

Additionally, let $\Sigma_{3}$ be the system obtained in Example \ref{example1}. $\Sigma_{3}$ is output observable because

\begin{equation*}
T^{(3)}_{L} = \left( \begin{array}{ccccccccccc} 
2 & 2 & 1 & 2 & 2 & 0 & 0 & 0 & 0 & 0 & 0 \\
0 & 2 & 0 & 1 & 1 & 2 & 2 & 0 & 0 & 0 & 0 \\
1 & 2 & 0 & 0 & 1 & 1 & 1 & 2 & 2 & 0 & 0 \\
1 & 0 & 0 & 0 & 1 & 0 & 1 & 1 & 1 & 2 & 2 \end{array} \right),
\end{equation*}
which also has full rank. Therefore, since $\mathcal{C}$ is GDP, the convolutional code $\mathcal{C}_{3}$ is also GDP. 
\end{example}

Since reachability, observability, and output observability are invariant under group actions on the parity and codeword vectors, if we start with a reachable, observable, and output observable I/S/O representation $\Sigma$ over $\mathbb{F}$, applying the group transformations described above will yield a reachable, observable, and output observable I/S/O representation $\Sigma_{i}$, which allows us to construct observable and GDP convolutional codes $\mathcal{C}_{i}$. The proposal developed in \cite{tomasvirtudes,tesisvirtudes,lieb,lieb2} involves obtaining $\Sigma$ by imposing that $[\Omega_{L+1}(A,C) \mid F_{L}]$ is a superregular matrix. Since these matrices have the property that all non-trivial minors are non-zero, the matrix $[\Omega_{L+1}(A,C) \mid F_{L}]$ has full rank. Thus, if $B$ is full rank, $\Sigma$ is reachable, and if $\Omega_{L}(A,C)$ has full rank, then $\Sigma$ is observable. Consequently, the associated convolutional code is observable, GDP, and can recover $n-k$ erasures out of $n$ symbols.

\begin{example}
As an example, consider $\Sigma$ as a linear system such that $T_{L}$ is a superregular matrix. We take the linear system $\Sigma=(A,B,C,D)$ as given in \cite{lieb,lieb2}, where a GDP (5, 3, 2)-convolutional code with maximum delay $T = L = 1$ is constructed. The system is defined by:

\[
A = \frac{1}{a^{8}-1} \left( \begin{array}{cc} 
a^{64}-a^{112} & a^{128}-a^{240} \\
a^{104}-a^{48} & a^{232}-a^{112} 
\end{array} \right), 
B = \left( \begin{array}{ccc} 
1 & 0 & -a^{32}(a^{8}+1)  \\
0 & 1 & a^{16}(a^{16}+a^{8}+1)  
\end{array} \right),
\]

\[
C = \left( \begin{array}{cc} 
a^{8} & a^{16} \\
a^{16} & a^{32} 
\end{array} \right), 
D = \left( \begin{array}{ccc} 
a & a^{2} & a^{4}  \\
a^{2} & a^{4} & a^{8} 
\end{array} \right),
\]

where $a$ is a primitive element of $\mathbb{F}$, and we consider $\mathbb{F}_{p^{N}}$ with $N > 330$. In this case, the matrix $T_{L}$
is superregular. Thus, the convolutional code associated with $\Sigma$ is reachable, observable, and GDP. 

Now consider the invertible matrix 
$$Q = \left( \begin{array}{ccc} 
a & a & a  \\
0 & a & a \\
0 & 0 & a 
\end{array} \right),$$
If we compute the system $\Sigma_{2} = (A_{2}, B_{2}, C_{2}, D_{2}) = (A, BQ, C, DQ)$, then:

\[
A_{2} = \frac{1}{a^{8}-1} \left( \begin{array}{cc} 
a^{64}-a^{112} & a^{128}-a^{240} \\
a^{104}-a^{48} & a^{232}-a^{112} 
\end{array} \right), 
B_{2} = \left( \begin{array}{ccc} 
a & a & -a^{41}-a^{33}+a  \\
0 & a & a^{33}+a^{25}+a^{17}+a  
\end{array} \right),
\]

\[
C_{2} = \left( \begin{array}{cc} 
a^{8} & a^{16} \\
a^{16} & a^{32} 
\end{array} \right), 
D_{2} = \left( \begin{array}{ccc} 
a^{2} & a^{3}+a^{2} & a^{5}+a^{3}+a^{2}  \\
a^{3} & a^{5}+a^{3} & a^{9}+a^{5}+a^{3} 
\end{array} \right),
\]
Since $B_{2}$ is of full rank, $\Sigma_{2}$ is reachable. It is also observable, and thus can be considered as a minimal I/S/O representation of a convolutional code $\mathcal{C}_{2}$ that is not equivalent to $\mathcal{C}$. Additionally, since $T_{L}^{2}$ is of full rank, $\mathcal{C}_{2}$ is GDP.
\end{example}

\section{Some Conjectures on MDP Construction of Convolutional Codes}

Regarding the distance properties, $\Sigma$ is an I/S/O representation of an MDP convolutional code if and only if each minor of $F_{L}$, which is not trivially zero, is non-zero [cf. Theorem 2.4., \cite{HUT}]. Here, it is clear that $D \neq O$ in order that $\Sigma$ can be the I/S/O of a MDP convolutional code. Thus, if $F_{L}$ is superregular, the associated convolutional code is an MDP code \cite{tomasvirtudes,tesisvirtudes,lieb,lieb2, ALMEIDA2013,ALMEIDA2016} . A natural question arises as to whether this property is invariant under group actions. The invariance of superregularity under group transformations might require additional conditions on the matrices used to perform the action. For this reason, we aim to propose some conjectures, starting with the necessary condition for the MDP property of a convolutional code to remain invariant under the transformations. Specifically, $F_{L}$ must be of full rank.

\begin{proposition}
\label{mainmdp}
Let $\Sigma=(A,B,C,D)$ be a minimal I/S/O representation of an MDP convolutional code defined over $\mathbb{F}$ with $D \neq O$. We consider the following group actions:
\begin{enumerate}
\item[i)] Group actions on the parity vector: $\Sigma=(A,B,C,D) \mapsto \Sigma_{2}= (A,BQ,C,DQ)$ for some invertible matrix $Q \in \mathbb{F}^{k \times k}$.
\item[ii)] Group actions on the information vector: $\Sigma=(A,B,C,D) \mapsto \Sigma_{3}= (A,B,H^{-1}C,H^{-1}D)$ for some invertible matrix $H \in \mathbb{F}^{(n-k) \times (n-k)}$.
\end{enumerate}
Let $F^{(i)}_{L}, i=2,3$ be the corresponding matrices, defined as
\begin{equation*}
F^{(i)}_{L}= \left(\begin{array}{ccccc} 
 D_{i} & O & O & O & O \\
 C_{i}B_{i} & D_{i} & O & O & O \\
 C_{i}A_{i}B_{i} & C_{i}B_{i} & D_{i} & O & O \\
 \vdots & & \vdots & \vdots & \vdots \\
 C_{i}A_{i}^{L-1}B_{i} & C_{i}A_{i}^{L-2}B_{i} & \dots & C_{i}B_{i} & D_{i} \end{array} \right),
\end{equation*}
Then, $F^{(i)}_{L}$ is full rank. 
\end{proposition}
\begin{proof}
We start that since $\mathcal{C}$ is MDP, then $F_{L}$ is full rank.\\
(i)  We have:
\begin{equation*}
F^{(2)}_{L} = 
\left(\begin{array}{ccccc} 
 D &  &  & & \\
 CB & D & \vdots & \vdots & \\
 CAB & CB & D & \vdots & \\
 \vdots & & \vdots & & \\
 CA^{L-1}B & CA^{L-2}B & \dots & CB & D \end{array} \right) \cdot \left(\begin{array}{ccccc} 
 Q &  &  & & \\
 & Q & \vdots & \vdots & \\
 & & Q & \vdots & \\
 \vdots & & \vdots & & \\
 & & \dots & & Q \end{array} \right),
\end{equation*}
which implies that $\text{rank}(F^{(2)}_{L}) = \text{rank}(F_{L})$. Therefore, if $\mathcal{C}$ is MDP, the necessary condition for $\mathcal{C}_{2}$ to be MDP is satisfied.\\
(ii) Similarly, since 
\begin{equation*}
F^{(3)}_{L} =
\left(\begin{array}{ccccc} 
H^{-1} & O &  &  & \\
O & H^{-1} & O & \vdots & \vdots \\
O & O & H^{-1} & O & \vdots \\
\vdots & \vdots & & \vdots & \\
O & O & O & \dots & H^{-1} \end{array} \right) \cdot \left(\begin{array}{ccccc} 
 D &  &  & & \\
 CB & D & \vdots & \vdots & \\
 CAB & CB & D & \vdots & \\
 \vdots & & \vdots & & \\
 CA^{L-1}B & CA^{L-2}B & \dots & CB & D \end{array} \right),
\end{equation*}
we obtain that $\text{rank}(F^{(3)}_{L}) = \text{rank}(F_{L})$. Therefore, if $\mathcal{C}$ is MDP, the necessary condition for $\mathcal{C}_{3}$ to be MDP is satisfied.
\end{proof}
\begin{example}
Let $\Sigma$ and $\Sigma_{2}$ be the systems given in Example \ref{example1}. We can verify that since $F_{L}$ has full rank, then $F^{(2)}_{L}$ is also full rank, where
\begin{equation*}
F_{L} = 
\left(\begin{array}{cccccccc} 
 1 & 1 & 0 & 0 & 0 & 0 & 0 & 0 \\
 2 & 2 & 1 & 1 & 0 & 0 & 0 & 0 \\
 0 & 2 & 2 & 2 & 1 & 1 & 0 & 0 \\
  0 & 2 & 0 & 2 & 0 & 2 & 1 & 1 
\end{array} \right)  \textrm{ and } F^{(2)}_{L} = 
\left(\begin{array}{cccccccc} 
 2 & 0 & 0 & 0 & 0 & 0 & 0 & 0 \\
 1 & 0 & 2 & 0 & 0 & 0 & 0 & 0 \\
 2 & 1 & 1 & 0 & 2 & 0 & 0 & 0 \\
  2 & 1 & 2 & 1 & 1 & 0 & 2 & 0 
\end{array} \right).
\end{equation*}

Furthermore, let $\Sigma_{3}$ be the system given in Example \ref{example1}. We can verify that $F^{(3)}_{L}$ has full rank, where
\begin{equation*}
 F^{(3)}_{L} = 
\left(\begin{array}{cccccccc} 
 2 & 2 & 0 & 0 & 0 & 0 & 0 & 0 \\
 1 & 1 & 2 & 2 & 0 & 0 & 0 & 0 \\
 0 & 1 & 1 & 1 & 2 & 2 & 0 & 0 \\
0 & 1 & 0 & 1 & 1 & 1 & 2 & 2
\end{array} \right).
\end{equation*}
\end{example}
\begin{remark}
If we also want the codes to be non-equivalent, it is likely that we will need to impose additional conditions on the matrices $Q$ and $H$.
\end{remark}
A sufficient condition for the proposed transformations to preserve the MDP property would require that the condition of every trivially non-zero determinant of $F_{L}$ being non-zero remains invariant under the group actions. 
Although we have shown that the maximum rank of $F_{L}$ is preserved, proving that this result holds for all non-trivially non-zero minors in $F_{L}$ is not straightforward. The first condition that is necessary is already determined by the singularity of the matrices $Q$ and $H$. However, multiplying an invertible matrix by a superregular matrix does not generally preserve the superregularity of the product. Only in case the invertible matrix is diagonal we are sure that it preserves superregularity by matrix multiplication.

\begin{proposition}
\label{mainmdp2}
Let $\Sigma$ over $\mathbb{F}$ be a minimal I/S/O representation of an MDP convolutional code. The MDP property is invariant under:
\begin{enumerate}
\item[i)] Group actions on the parity vector: $\Sigma= (A,B,C,D) \mapsto \Sigma_{2}=(A,BQ,C,DQ)$ for some diagonal invertible matrix $Q \in \mathbb{F}^{k \times k}$.
\item[ii)] Group actions on the information vector: $\Sigma=(A,B,C,D) \mapsto \Sigma_{3}=(A,B,H^{-1}C,H^{-1}D)$ for some diagonal invertible matrix $H \in \mathbb{F}^{(n-k)\times (n-k)}$.
\end{enumerate}
\end{proposition}
\begin{proof}
(i) Let $N$ be a $k\times n$ matrix and $Q$ an invertible diagonal matrix of size $n$ with diagonal elements $d_1,\hdots,d_n$. Suppose $N(I,J)$ is the $l\times l$ submatrix of $N$ corresponding to rows with indices in $I$ and columns with indices in $J$. Then, the $l\times l$ submatrix $(N\cdot Q)(I,J)$ of $N\cdot Q$ corresponding to the same rows and columns is computed by multiplying the rows of $N$ with indices in $I$ with the columns of $Q$ with indices in $J$. Thus, the $j$th column of $(N\cdot Q)(I,J)$ is equal to the $j$th column of $N(I,J)$ multiplied by the $j$th element of the diagonal of $Q$. Therefore, $det((N\cdot q)(I,J))=\prod_{j\in J}d_i M(I,J)$. So, $det((N\cdot Q)(I,J))\neq 0$ if and only if $det(N(I,J))\neq 0$. The result follows from this fact.\\
(ii) This follows from an argument similar to the above but with left multiplication.
\end{proof}
\begin{conjetura}
\label{mainmdp3}
The largest subgroup of the general linear group (of the appropriate size) that preserves the MDP property under the group actions defined in Proposition \ref{mainmdp2} is the subgroup of invertible diagonal matrices.
\end{conjetura}



\section{Conclusions}
Convolutional codes are deeply connected to several areas of mathematics. These connections enable us to analyze the dynamics of encoders, construct convolutional codes with specific properties, and develop more efficient decoding methods.

In this paper, we have focused on utilizing group actions and the I/S/O representation of convolutional codes to construct new classes of observable and GDP convolutional codes. Also, we give some results in order to construct MDP convolutional codes.

As part of future work, once the I/S/O representations of convolutional codes over modular integer rings \cite{NAPPring,angelring} are fully developed, it will be natural to explore the application of the proposed construction to the decoding problem for convolutional codes over these rings.



\section*{Disclosure statement}

All authors declare that they have no conflicts of interest.

\section*{Funding}
This work was conducted within the framework of the project \emph{Algebraic Methods for the Recovery, Correction, and Security of Digital Information (MARCSID)}, funded by the Spanish Ministry of Science and Innovation under the State Plan for Scientific and Technical Research and Innovation 2021-2023, with grant number TED2021-131158A-I00.

\bibliographystyle{tfnlm}
\bibliography{referencesfinal}
\end{document}